\documentclass[11point]{amsart}
\usepackage{amsmath}
\usepackage{amscd}

\usepackage{amssymb,amsmath,amscd,epsfig,amsthm,enumerate,import}

\newtheorem{theorem}{Theorem}[section]
\newtheorem{proposition}[theorem]{Proposition}
\newtheorem{lemma}[theorem]{Lemma}
\newtheorem{corollary}[theorem]{Corollary}
\theoremstyle{definition}

\theoremstyle{remark}
\newtheorem{remark}[theorem]{Remark}

\numberwithin{equation}{section}
\newcommand{\be}{\begin{equation}}
\newcommand{\ee}{\end{equation}}

\newcommand{\bbC}{{\mathbb C}}
\newcommand{\bbZ}{{\mathbb Z}}
\newcommand{\bbR}{{\mathbb R}}
\newcommand{\bbP}{{\mathbb P}}

\newcommand{\calG}{{\mathcal G}}

\newcommand{\calK}{{\mathcal K}}
\newcommand{\calL}{{\mathcal L}}

\newcommand{\calH}{{\mathcal H}}

\newcommand{\calF}{{\mathcal F}}

\newcommand{\calO}{{\mathcal O}}

\newcommand{\frakh}{{\mathfrak h}}
\newcommand{\frakg}{{\mathfrak g}}
\newcommand{\frakn}{{\mathfrak n}}
\newcommand{\frakt}{{\mathfrak t}}
\newcommand{\frakG}{{\mathfrak G}}

\DeclareMathOperator{\tr}{Tr}

\newcommand{\inner}[2]{\langle#1,#2\rangle}

\newcommand{\norm}[1]{\lVert#1\rVert}

\newcommand{\isom}{\cong}

\newcommand{\bbJ}{{\mathbb J}}

\begin{document}

\title{The exponential map of the complexification of {\em Ham} in the real-analytic case.}
\author{D. Burns}
\address{Mathematics Department\\
University of Michigan\\Ann Arbor, Michigan 48109}
\email{dburns@umich.edu}
\thanks{D.B. supported in part by NSF grant DMS-0805877}
\author{E. Lupercio}
\address{CINVESTAV\\Departamento de Matem\'aticas\\Apartado Postal 14-740 07000\\
M\'exico, D.F., M\'exico}
\email{lupercio@math.cinvestav.mx}
\thanks{E.L. was supported in part by Fundaci\'{o}n Marcos Moshinsky, Instituto de Matem\'{a}ticas UNAM , Conacyt and by Laboratory of Mirror Symmetry NRU HSE, RF Government grant, ag. № 14.641.31.0001}
\author{A. Uribe}
\address{Mathematics Department\\
University of Michigan\\Ann Arbor, Michigan 48109}
\email{uribe@umich.edu}
\thanks{A.U. supported in part by NSF grant DMS-0805878.}

\date{\today}
\begin{abstract}
Let $(M, \omega, J)$ be a K\"ahler manifold and $\calK$ its group of Hamiltonian symplectomorphisms.
The complexification of $\calK$ introduced by Donaldson is not a group, only a ``formal Lie group".  However, it still
makes sense to talk about the exponential map in the complexification.    In this note 
we show how to
construct geometrically the exponential map (for small time), in case the initial data are real-analytic.  The construction
is motivated by, but does not use, semiclassical analysis.
\end{abstract}

\maketitle
\tableofcontents

\section{Introduction and statement of the results}
A real Lie algebra can be easily complexified by tensoring it with the complex numbers over the 
field of real numbers, and extending the Lie bracket bilinearly.  Complexifying a Lie group is a much more subtle 
problem, for which a solution does not always exist.  All compact Lie groups admit a 
complexification,  
but the proof of that is not trivial (see, for example \cite{Zhelobenko}, \S 106).

Let $(M,\omega, J)$ be a K\"ahler manifold, 
and let $\calK$ denote the group of Hamiltonian 
symplectomorphisms of $(M,\omega)$, with Lie algebra $C^\infty(M,\bbR)/\bbR$.   
$\calK$ is  known to be ``morally" an infinite-dimensional analogue of  a compact group, 
and one can wonder 
whether it has 
a complexification.  (One can rigorously prove
that $\calK$ is a diffeological Lie group, \cite{IZ}.)
From a physical point of view this would correspond to finding a sensible way 
to associate a dynamical system to a complex-valued Hamiltonian, $h:M\to\bbC$, 
in a manner that extends the notion of Hamilton flow in case $h$ is real-valued.

This issue was raised and taken on by Donaldson in a series of papers beginning with
\cite{Do}, \cite{Do2}, in connection with a set 
of lovely problems in K\"ahler geometry, and has generated a lot of research.  From this point of view, the interest in finding 
a complexification of {\em Ham} is based on the fact, discovered independently by Atiyah and Guillemin and 
Sternberg, that if a compact Lie group acts in a Hamiltonian fashion on a K\"ahler manifold then the the action extends 
to the complexified group in an interesting way that can be understood.  Donaldson pointed out that {\em Ham} acts 
on the infinite-dimensional space of K\"ahler potentials of $(M, \omega_0)$, and the moment map 
is a Hermitian scalar curvature.  The appeal of extending the finite-dimensional picture 
to this infinite-dimensional setting of K\"ahler metrics is that it was hoped to 
lead to a way of constructing
extremal metrics.  In this note we will not enter into the connections with K\"ahler geometry, 
other than to point out here that
exponentiating purely-imaginary Hamiltonians leads to geodesics in the space of K\"ahler
potentials.  

Our focus is on the following issue:  Donaldson has put forward a notion of ``formal Lie group"  to conceptualize his notion of a formal 
complexification of {\em Ham}.  Briefly, a formal Lie group with Lie algebra $\frakG$ is a 
``manifold" (the notion is of interest
only in infinite dimensions), $\calG$, together with a trivialization of its tangent bundle of the
form $T\calG \cong \calG\times\frakG$, such that the map $\frakG\ni h\mapsto $corresponding
vector field $h^\sharp$ on $\calG$ is a Lie algebra homomorphism (with respect to the 
commutator of vector fields).  The vector fields $h^\sharp$ should be thought of as ``left-
invariant", though no group structure on $\calG$ exists.  (Again this definition can be made completely
rigorous in the language of diffeologies, \cite{IZ}.)

\medskip
In the present case
$\frakG = C^\omega(M,\bbC)$ and the exponential map in our title refers to
the problem of constructing the flow of the fields $h^\sharp$. 
Corollary \ref{CorMain} states that the family of
diffeomorhpisms $\{f_t\}$ there exponentiate the complex-valued Hamiltonian $h$, in one of Donaldson's models for the complexification of $\calG$ (see \cite{Do2} \S 4).
 
 \medskip
Our construction is based on the following existence result:

\begin{proposition}\label{conditions}
Let $(M, J, \omega)$ be a real analytic K\"ahler manifold of real dimension $2n$.
There exists a holomorphic complex symplectic manifold $(X,I)$ of complex dimension $2n$
and an inclusion $\iota: M\hookrightarrow X$ such that $\iota^*\Omega = \omega$, and
with the following additional structure:
\begin{enumerate}
\item An anti-holomorphic involution $\tau: X\to X$ whose fixed point set is the image of $\iota$ and such that $\tau^*\Omega = \overline{\Omega}$.
\item A holomorphic
projection $\Pi: X\to M$, $\Pi\circ\iota = \mathrm{Id}_M$, whose fibers are holomorphic Lagrangian submanifolds.
\end{enumerate}
\end{proposition}

The local existence is simple:  We take $X$ to be a neighborhood of the diagonal in $M\times M$, 
with the complex structure $I=(J, -J)$.  $\Omega$ is the holomorphic extension of $\omega$ and 
$\tau(z,w) = (w,z)$.  Finally, the projection is simply $\Pi(z,w) = z$.  However, there
exist natural complexifications that make our results below much more global in some cases.  

For example, if $M$ is a generic coadjoint orbit of a compact simply connected
Lie group $G_0$ (one through the interior of a Weyl chamber)  then, one can
take for $X$ the orbit of the complexification, $G$, of $G_0$ through the same element (see \S 3 for details).

\medskip
To describe our results we need some notation.
Given a function $h:M\to\bbC$ whose real and imaginary parts are real analytic, there is
a holomorphic extension $H:X\to \bbC$  perhaps
only defined near $\iota(M)$, but we will not make a notational distinction between $X$ and 
such a neighborhood, as our results are local in time.

\medskip
Let $h$ and $H$ be as above.  
The fibers of $\Pi$ are the leaves of a holomorphic foliation, $\calF$, of $X$. 
Denote by $\Phi_t: X\to X$ the Hamilton
flow of $\Re H$ (where $\Re H$ is the real part of $H$) with respect to the real part of $\Omega$.
We denote by $\calF_t$
the image of the foliation $\calF$ under $\Phi_t$ (so that $\calF_0 = \calF$).
 We assume
that there exists an open interval $E\subset\bbR $ containing the origin such that for all $t\in E$, the leaves of 
$\calF_t$ are the fibers of a projection $\Pi_t: X\to M$.
We will denote
\[
\calF_t^x := \Pi_t^{-1}(x)
\]
the fiber of $\calF_t$ over $x$.

\begin{theorem}\label{Main}
Let $E\subset\bbR $ be an open set as above.  
Let $\phi_t:M\to M$ be defined by 
\begin{equation}\label{lapapa}
\phi_t := \Pi_t\circ\Phi_t\circ\iota.
\end{equation}
Then, $\forall t\in E$:
\begin{enumerate}
\item There is a complex structure $J_t:TM\to TM$ such that $J_t\circ d\Pi_t = d\Pi_t\circ I$.
\item $\forall x\in M\  \dot{J_t}(x) := \frac{d\ }{dt}J_{t,x}: T_xM\to T_xM$ satisfies
the equation
\begin{equation}\label{jDotEq}
 \dot{J_t} = - \calL_{\Xi^{\omega}_{\Re h} + J_t\left(\Xi^{\omega}_{\Im h}\right)} J_t, \quad
 J_t|_{t=0} = J.
\end{equation}
\item  The infinitesimal generator of $\phi_t$ is
\begin{equation}\label{main}
\dot{\phi}_t\circ \phi_t^{-1}= \Xi^{\omega}_{\Re h} + J_t\left(\Xi^{\omega}_{\Im h}\right)
\end{equation}
where $\Xi^\omega_{\Re h}$, denotes the Hamilton vector field of $\Re h$ with respect to $\omega$, etc.
\end{enumerate}

\end{theorem}

\begin{remark}\label{CK}
Conditions (2) and (3) (together with the initial condition $\phi_0 = \mathrm{Id}_M$) characterize the family $\{\phi_t\}$.
It is in fact not hard to see that in local real-analytic coordinates $u_j$ on $M$, 
these conditions show that the family solves a system of first-order non-linear PDEs of the form
\[
\dot{U}_j(u,t) = F_j(U, U_{u_k})
\]
with $U_j(u, 0) = u_j$ and where the $F_j$ are real-analytic.  Therefore local existence and uniqueness follow
from the Cauchy-Kowalewsky theorem.
\end{remark}

\begin{remark}\label{Jdot}
In light of (\ref{main}), condition (\ref{jDotEq}) is equivalent to $\phi_t^*J_t = J$, i.e.
\begin{equation}\label{actuallyShow}
\forall t\qquad
\phi_t: (M, J)\to (M,J_t)\quad\text{is holomorphic,}
\end{equation}
that is, $J_t\circ d\phi_t = d\phi_t\circ J$.  Indeed
 \begin{equation}\label{}
\frac{d\ }{dt} \phi_t^*\left(J_t\right) = \phi_t^*\left(\dot J_t\right) + \phi_t^* \left(\calL_{\Theta_t}J_t\right),
\end{equation}
where
\begin{equation}\label{Xi}
\Theta_t = \Xi^{\omega}_{\Re h} + J_t\left(\Xi^{\omega}_{\Im h}\right)
\end{equation}
and the Lie derivative $\calL_{\Theta_t}J_t$ is taken in the usual way (with $t$ fixed).
In \S 3.1 we will write equation (\ref{jDotEq}) in local coordinates.
\end{remark}

\begin{remark}\label{CompactGroup}
Suppose $G$ is a compact Lie group acting on $M$ in a Hamiltonian fashion and preserving $J$.  Then, 
the action extends to a holomorphic action to the complexification $G_\bbC$ (c.f. \cite{GS}, \S 4).  The 
extended action is as follows:  If
$a,b :C^\infty(M)\to\bbR$ are two components of the moment map of the $G$ action, then the infinitesimal action
corresponding to $a+ib$ is the vector field $\Xi^\omega_a + J(\Xi^\omega_b)$.  The corresponding
one-parameter group of diffeomorphisms,  $\varphi_t: M\to M$, satisfies (2) and (3) of Theorem \ref{Main},
with $J_t = J_0$ for all $t$.  By the uniqueness part of the previous remark, we must have $\varphi_t = \phi_t$.
In other words, our construction is an extension of the process of complexifying the action of a compact
group of symmetries of $(M, \omega, J)$.
\end{remark}

We now explain the geometry behind the construction of $\phi_t$ summarized by (\ref{lapapa}).  
To find the image of $x\in M$ under $\phi_t$, one flows the 
leaf $ \calF_0^x=\Pi^{-1}(x)$ of the foliation $\calF=\calF_0$ by $\Phi_t$, and
intersects the image leaf with $M$.  In other words, (\ref{lapapa}) can be stated equivalently as:
\begin{equation}\label{lapapa2}
 \{ \phi_t(x)\} = \Phi_t\left(\Pi^{-1}(x)\right) \cap M .
\end{equation}
The definition of $\phi$ is summarized in the following figure, where $\calF_t^y := \Pi_t^{-1}(y)$:

\medskip

\begin{center}
\includegraphics[scale=.3]{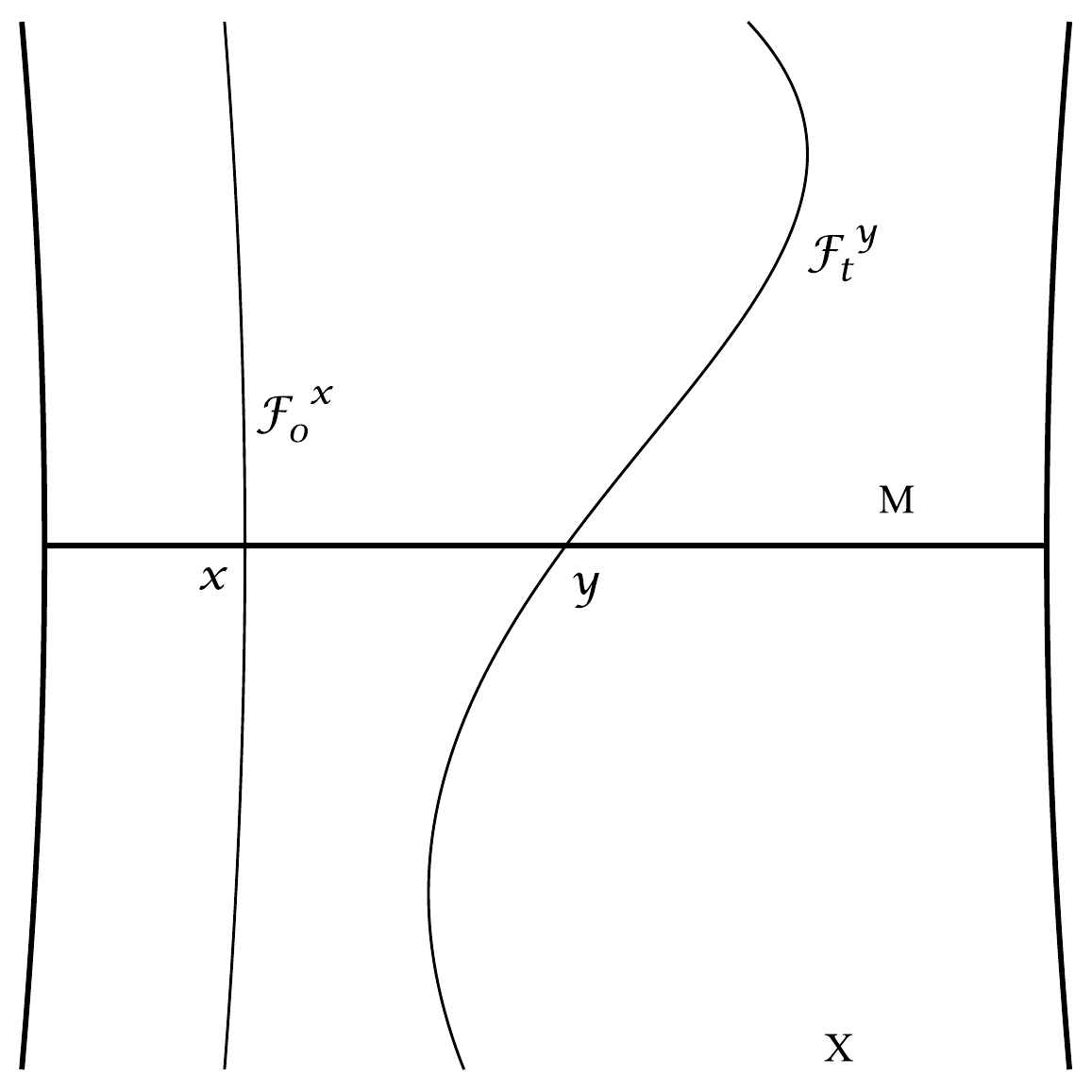}

\medskip
In this figure $M$ is represented by the horizontal segment.  Under $\Phi_t$ the fibers of the
foliation $\calF_0$ are transformed into the fibers of $\calF_t$, and $y = \phi_t(x)$.
\end{center}

\bigskip
As we will see, the following is an easy consequence of the previous result:
\begin{corollary}\label{CorMain}
Let $f_t := \Pi_0\circ \Phi_t\circ \iota:  M\to M$ and $\omega_t$ the symplectic
form defined by $\omega = f_t^*\omega_t$.  Then
\begin{equation}\label{fCoro}
\dot{f}_t\circ f_t^{-1}= \Xi^{\omega_t}_{\Re h\circ f_t^{-1}} + J_0\left(\Xi^{\omega_t}_{\Im h\circ f_t^{-1}}\right),
\end{equation}
where $\Xi^{\omega_t}_{\Re h\circ f_t^{-1}}$, denotes the Hamilton vector field of $\Re h$ with respect to $\omega_t$, etc.
\end{corollary}
It is not hard to check that $f_t = (\phi_{-t})^{-1}$, therefore in the (relatively rare) cases when $\phi_t$ is a
one-parameter subgroup of diffeomorphisms one has $f_t = \phi_t$.

\medskip
The present construction is motivated by semiclassical analysis.  Ignoring (possibly
catastrophic) domain issues, the notion of the exponential of a non-Hermitian {\em quantum} 
Hamiltonian is clear.
(For example, if $M$ is compact and Planck's constant is fixed, this amounts to exponentiating
a matrix.)  Therefore, a very natural approach to exponentiating a non-Hermitian
classical Hamiltonian is to first quantize it, exponentiate it on the quantum side, and then
take the semiclassical limit.  This approach has been developed by Rubinstein and Zelditch,
\cite{RZ1, RZ2}, and raises a number of interesting but difficult analytical questions.  
In our construction, we bypass these analytic difficulties by considering the following
geometric remnants
of quantization:  Each leaf of the foliation $\calF_0$ corresponds to a quantum state
(element of the projectivization of the quantum Hilbert space)
represented by a {\em coherent state} centered at a point on the leaf, that is,
an element in the Hilbert space that semiclassically concentrates at the intersection
of the leaf with $M$.  The fact that $\Pi$ is holomorphic
says the the coherent states are associated to the metric of $(M,J,\omega)$.  
On the quantum side, the evolution of a coherent
state remains a coherent state, whose Lagrangian is simply the image of the one at time
$t=0$ by the complexified classical flow.  As explained above, our maps $\{\phi_t\}$
simply follow the evolution of the real center.


We mention that the paper \cite{GrSc} by Graefe and Schubert contains a very 
clear and detailed account of the case when $M$ is 
equal to $\bbR^{2n}\cong\bbC^n$, $h$ is a quadratic 
complex Hamiltonian, and the Lagrangian foliations are given by complex-linear positive subspaces,
corresponding to standard Gaussian coherent states with possibly complex centers.  
By explicit calculations on both quantum and classical sides, they show that the the evolution
of a coherent state centered at $x$ is another coherent state whose center may be complex, but
that represents the same quantum state as a suitable Gaussian coherent state centered at 
$\phi_t(x)$. 

Finally, we mention that other authors have examined the situation treated here from different points of view. Hall and Kirwin \cite{HK}, developing an earlier observation of Thiemann \cite{TT}, use imaginary time dynamics to alter the complex structure on the classical phase space $T^*M$, generalizing Grauert tube constructions. Kirwin, Mour\~ao, Nunes \cite{KMN} used complexified dynamics, especially on toric varieties, to study the relation between real and complex polarizations in geometric quantization. There, the degenerations we describe here in terms of $\phi_t, J_t$ above, are reflected in the degeneration of positivity in K\"ahler polarizations, leading to ``islands" of non-positivity for the Hilbert structure. The idea that one can either change the complex structure, or equivalently, the K\"ahler metric, seems to go back to Semmes \cite{SS}, rediscovered and set in a very interesting context in K\"ahler geometry, by Donaldson \cite{Do2}. It was Donaldson's casting of the situation that was our point of departure. It is an interesting question as to whether further advances in K\"ahler geometry and quantization of non-Hermitian Hamiltonians have deeper connections developing the more formal relations exploited here.

%
%
%
%
%
%
%
%
%
%
%
%
%
%
%
%
%
%
%
%

\subsection{Relation with geodesics in the space of K\"ahler potentials}

We start again with $(M,\omega)$ as above.  Consider the space
\[
\calH := \{ a:M\to\bbR\;;\; \omega_a := \omega + i\bar\partial\partial a > 0\}/\bbR
\]
of K\"ahler potentials for K\"ahler forms $\omega_a$, in the same cohomology class as $\omega$.

\begin{lemma}
Given $h= F+iG: M\to\bbC$, let $f_t$ be its exponential.   Let $\omega_t$ be the
symplectic form defined by $f_t^*\omega_t = \omega$, and write
$\omega_t = \omega + i\bar\partial\partial a_t$, where the $a_t$ are taken modulo constants.
Then 
\begin{equation}\label{unopuntocinco}
f_t^*\dot a_t = 2G.
\end{equation}
\end{lemma}
\begin{proof}
Differentiating
\[
f_t^*[\omega_t] = \omega,
\text{where} \ \omega_t = \omega + i\bar\partial\partial a_t
\]
with respect to time,  we get
\[
0 = f_t^*[\calL_{V_t}\omega_t + \dot\omega_t] =  f_t^*\left(d(V_t\rfloor\omega_t )+ 
i\bar\partial\partial \dot a_t\right)
\]
where $V_t = \dot{f}_t\circ f_t^{-1}$.
Therefore $d(V_t\rfloor\omega_t )+ i\bar\partial\partial \dot a_t=0$.
Recalling that 
$ \dot{f}_t\circ f_t^{-1} = \Xi^{\omega_t}_{F\circ f_t^{-1}} + J_0\left(\Xi^{\omega_t}_{G
\circ f_t^{-1}}\right)$, one computes:
\[
d(V_t\rfloor\omega_t ) = d^2(F\circ f_t^{-1}) - 
d\left(\omega_t(\Xi^{\omega_t}_{G\circ f_t^{-1}}\, , \,J_0(\cdot))\right) = 
-d(G\circ f_t^{-1})\circ J_0.
\]
This leads to the identity
\[
i\bar\partial\partial \dot a_t = d(G\circ f_t^{-1})\circ J_0.
\]
However $d(G\circ f_t^{-1})\circ J_0 = 2i\bar\partial\partial (G\circ f_t^{-1})$, 
therefore $i\bar\partial\partial \dot a_t = 2i\bar\partial\partial (G\circ f_t^{-1})$.
Since we are working modulo constants, the result follows.
\end{proof}

\medskip

$\calH$ has a natural Riemannian metric:
\[
\norm{\delta a}^2 = \int_M |\delta a|^2 d\mu_a.
\]

\begin{corollary}
Let $f_t$ be the exponential of a purely-imaginary Hamiltonian, $h = iG$, 
$G:M \to \bbR$.
Then the curve of potentials $a_t$
of the metrics $(\omega_t, J)$ is the geodesic on $\calH$, with initial condition
\[
\dot a_0 = 2G.
\]
\end{corollary}
\begin{proof}
The geodesic equation turns out to be
\[
\ddot{a} = -\frac{1}{2}|\nabla^a \dot a|_a^2.
\]
Differentiating (\ref{unopuntocinco}) with respect to time, get
\[
0 = \frac{d\ }{dt} f_t^*\dot a_t = f_t^*[\ddot a_t + d(\dot a_t)(V_t)],
\]
and therefore
\[
\ddot a_t = - d(\dot a_t)(V_t) = -d(\dot a_t)(\nabla^t(G\circ f_t^{-1}) = 
-d(\dot a_t)(\nabla^t\frac{1}{2} \dot a_t) = -\frac{1}{2}|\nabla^t \dot a_t|^2_t.
\]
\end{proof}

{\sc Note:}  The factor of two in the equation $\dot a_0 = 2G$ can be gotten
rid of by letting
\[
\omega_a = \omega_0 +2 i\bar\partial\partial a.
\]

\section{The proofs of Theorem \ref{Main} and Corollary \ref{CorMain}.}

Let $\xi$ denote the infinitesimal generator of $\Phi_t$.  
The following is easy to check in a local trivialization of the foliation of $X$
by fibers of $\Pi_t$:

\begin{lemma}\label{Cool}
Fix $t\in E$ and $x\in M$, and let $y=\phi_t(x)$ and $\xi$ denote
the infinitesimal generator of the 
flow $\Phi$.  Then  $\dot{\phi_t}(x)\in T_yM$ is 
\begin{equation}\label{cool}
\dot{\phi_t}(x) = d\left(\Pi_t\right)_y(\xi_y).
\end{equation}
\end{lemma}
\begin{proof}  
Introduce coordinates $(u,v)$ in a neighborhood $U\subset X$ centered at $y$
so that $U\cap M$ is defined
by $v=0$ and the projection $\Pi_t$ is just $\Pi_t(u,v) = u$.  Note that, since $\Phi$ 
is a one-parameter local subgroup of diffeomorphisms, for $s$ small enough
\[
\phi_{t+s}(x) = \Pi_{t+s}\circ \Phi_s(y),
\]
For $s$ near zero denote, the map $\Phi_s$ in coordinates by
\[
\Phi_{s}(u,v) = (U(s,u,v), V(s,u,v))
\]
(in a smaller neighborhood of $y$).
For each $s$ 
the image of the $v$-axis under $\Phi_s$ locally parametrize the fiber  $\calF^y_{t+s}$,
namely
\[
v\mapsto (U(s, 0, v), V(s, 0,v)).
\]
Therefore we can write $\phi_{t+s}(x) = U(s, 0, v(s))$,
where $v(s)$ is defined implicitly by $V(s,0, v(s)) = 0$ and
$v(0) = 0$.  It follows that
\[
\dot\phi_t(x) = \dot{U}(0, 0, 0) + \frac{\partial U}{\partial v}(0,0,0)\cdot\dot{v}(0).
\]
However $\Phi_0$ is the identity, so that $U(0,u,v) = u$ and, therefore, 
$\frac{\partial U}{\partial v}(0,0,0) = 0$.
\end{proof}

To proceed further we will need some notation.  
We regard $X$ as a real manifold of real dimension $4n$ with an integrable
complex structure $I:TX\to TX$.  Let us write 
\[
\Omega = \omega_1 + i\omega_2
\]
for the real and imaginary parts of $\Omega$.  Thus the $\omega_j$ are real symplectic forms on $X$
and $M$ is $\omega_1$-symplectic and $\omega_2$-Lagrangian.  Let us write
\begin{equation}\label{re+iim}
H = F + iG
\end{equation}
for the real and imaginary parts of $H$.  Recall that, by definition, $\xi$ is the Hamilton field
of $F$ with respect to $\omega_1$.  

\begin{lemma}
$\Xi^\Omega_{2H}=\xi - iI(\xi)$ is the holomorphic vector field on $X$ associated to the Hamiltonian $2H$ with respect to the form $\Omega$.
Therefore $\Phi_t$ is a holomorphic automorphism of $(X,\Omega)$.
\end{lemma}
\begin{proof}
We first note that, since $\Omega$ is of type $(2,0)$, 
\begin{equation}\label{hamfh}
\omega_1\rfloor I\xi = -\omega_2\rfloor\xi\quad\text{and}\quad 
\omega_2\rfloor I\xi = \omega_1\rfloor\xi .
\end{equation}
From this it follows easily that
$\Omega\rfloor (\xi - iI(\xi)) = 2dH$.  For the final statement just use that $\Omega$ and $H$ are holomorphic.
\end{proof}

For future reference we note the relations
\begin{equation}\label{hamfi}
\omega_1\rfloor\xi =d\Re H\quad\text{and}\quad \omega_1\rfloor I(\xi) = - d\Im H
\end{equation}
that follow from (\ref{hamfh}).

\begin{lemma}\label{realcase}
Suppose that $h:=\iota^*H$ is real.  Then $\xi$ is tangent to $M$, and its restriction to $M$ is
the Hamilton field of $h$ with respect to $\omega$.
\end{lemma}
\begin{proof}
If $h$ is real then $\tau^*H = \overline H$ (by uniqueness of analytic continuation of $h$), so $\Re H$ is 
$\tau$-invariant.  Since $\tau^*\Omega = \overline{\Omega}$, 
$\omega_1$ is also $\tau$-invariant, and therefore
$\tau$ maps $\xi$ to itself and so $\xi$ has to be tangent to the fixed-point set of $\tau$.  For the second part just note that $\omega = \iota^*\omega_1$.
\end{proof}

\medskip\noindent
\underline{Proof of Theorem \ref{main}}.  We take one point at a time:

\smallskip
\noindent (1)  Since the fibers of $\Pi_t$ are the leaves of a holomorphic foliation, there is a well-defined
complex structure in the abstract normal bundle to the fibers.  The inclusion $\iota: M\hookrightarrow X$ realizes $M$
as a cross-section to the foliation and identifies $TM$ with the normal bundle to the foliation along $M$.
Therefore, it inherits a complex structure that makes $\Pi_t$ holomorphic.

\smallskip
\noindent(2)  We will actually show (\ref{actuallyShow}).
This follows from the interpretation of the structures $J_t$ as arising from the normal
bundle structures together with the fact that $\Phi_t$ is holomorphic, or can be checked directly as follows.
Let $v\in T_xM$, then $I(v) = J_0(v) + w$, where $w\in T_x\calF_0^x$.  Since $d\Phi_t$ is holomorphic, one has
\[
Id\Phi_t(v) = d\Phi_t( J_0 (v)) + d\Phi_t(w).
\]
But $d\Phi_t(w)\in T\calF_t$ since $\Phi_t$ maps fibers of $\calF_0$ to fibers of $\calF_t$.  Therefore,
the previous relation implies that 
\[
d(\Pi_t)( d\Phi_t( J_0 (v))) = d(\Pi_t)(Id\Phi_t(v)) = J_t d(\Pi_t)(d\Phi_t(v)),
\]
which precisely says that $\phi_t$ is holomorphic.

\smallskip
\noindent(3) Omitting the subscript $t$ for simplicity,  by Lemma \ref{Cool} we need to show that
\[
d\Pi_x(\xi) = \Xi_{\Re h} + \nabla\Im h
\]
where:
\begin{enumerate}
\item $ \Xi_{\Re h}$ is the Hamilton field of the real part of $h$ with respect to $\omega$, and
\item $\nabla\Im h$ is the gradient of the imaginary part of $h$ with respect to the metric $(\omega, J)$.
\end{enumerate}
By the previous lemma, if $h$ is real, $d\Pi(\xi_x)= \xi_x$ and there is nothing more to prove. 
Suppose now that $h$ is purely imaginary.  By the second relation in (\ref{hamfi}), $I(\xi)$
is the Hamilton field of $-G$  (see \ref{re+iim}), and by the lemma \ref{realcase} $I(\xi)$ is tangent to $X$ and its restriction to $X$ is the Hamilton field of $ih = \iota^*(-G+iF)$ with respect to $\omega$.  Therefore, in this case, we can write
\begin{equation}\label{}
-\Xi_{\Im H} = d\Pi (I(\xi)) = J\, d\Pi(\xi),
\end{equation}
and it suffices to apply $J$ to both sides to get the result.  The general case follows by $\bbR$-linearity of
the composition $h\mapsto H\mapsto \xi$, where the first arrow is analytic continuation.
This concludes the proof of Theorem \ref{Main}.

\bigskip
\noindent\underline{Proof of Corollary \ref{CorMain}}.
Recalling that $f_t = (\phi_{-t})^{-1}$, so we wish to compute
$\dot f_t \circ \phi_{-t}$.
Differentiating with respect to time the identity $f_t\circ \phi_{-t}(x) = x$, we get
\[
\dot{f}_t \circ \phi_{-t}(x) = d(f_t)(\dot{\phi}_{-t}(x)) = 
d(f_t) \left[\Xi^\omega_{\Re h} + J_{-t}(\Xi^\omega_{\Im h})\right]_{\phi_{-t}(x)}
\]
Now it is not hard to check that 
$d(f_t)\left[\Xi^\omega_{\Re h}\right] = \Xi^{\omega_t}_{\Re h\circ f_t^{-1}}$ and
that $df_t\circ J_{-t} = J_0\circ df_t$ (using that $\phi_{-t}: (M,J_0)\to (M, J_{-t})$
is holomorphic).  Therefore,
\[
\dot{f}_t \circ \phi_{-t}(x) = \left[\Xi^{\omega_t}_{\Re h\circ f_t^{-1}} + 
J_0(\Xi^{\omega_t}_{\Im h\circ f_t^{-1}})\right]_{f^{-1}_{t}(x)}.
\]

\section{Further remarks and examples}

\subsection{The equation for $\dot{J}_t$}


Let $\phi_t: M\to M$ be the exponential of 
$h= F +iG: M\to\bbC$, i.e.
$
\dot{\phi_t}\circ \phi_t^{-1} = \Xi_F + J_t\Xi_G = \Theta_t.
$
We have shown that, for each $t$,
\begin{equation}\label{it'sHolo}
\phi_t: (M, J) \to (M, J_t)
\end{equation}
is holomorpic, and, therefore,
the complex structures $J_t$ satisfy the equation
\begin{equation}\label{jtEqn}
\dot J_t = -\calL_{\Theta_t}J_t.
\end{equation}
We want to make this equation explicit in local 
coordinates $(x_1,\ldots , x_{2n}) $.  
Let the components of $\Theta_t$ be
\begin{equation}\label{}
\Theta_t = \langle \Theta_t^1,\ldots ,\Theta_t^{2n}\rangle = \Xi_F +J_t\Xi_G.
\end{equation}
Then, if $\phi_t = ( \phi_t^1,\cdots , \phi_t^{2n})$ is the flow in coordinates, the
equations of motion are
\begin{equation}\label{}
\frac{d\ }{dt}{\phi_t}^j(x) = \Theta_t^j(\phi_t(x)).
\end{equation}
 All tangent
spaces are identified with $\bbR^{2n}$.  Fix an initial condition $x_0\in M$, and let
$\bbJ_t(x)$ be the {\em matrix} of the complex structure at time $t$ and at $T_xM$.  Let
\[
K_t = \text{the matrix of }\  d(\phi_t)_{x_0}: T_{x_0}M\to T_{\phi_t(x_0)}M.
\]
Then that (\ref{it'sHolo}) is holomorphic becomes 
\begin{equation}\label{}
K_t\,\bbJ_0 (x_0)= \bbJ_t(\phi_t(x_0))\,K_t
\end{equation}
From this, it follows that
\begin{equation}\label{}
\frac{d\ }{dt} \bbJ_t(\phi_t(x_0)) = [\dot K_t K_t^{-1}\,,\,\bbJ_t(\phi_t(x_0))]
\end{equation}
(matrix commutator).  To continue, note that
\begin{equation}\label{}
K_t = \begin{pmatrix}
\frac{\partial \phi^i_t}{\partial x_j}
\end{pmatrix}_{|_{x=x_0}}
\end{equation}
where $i$ is the row index and $j$ the column index.
To compute $\dot K_t$ we use the equations of motion and the fact that the partials
with respect to $x$ and with respect to $t$ commute:
\[
\frac{\partial^2 \phi_t^i}{\partial t\partial x_j} = \frac{\partial\ }{\partial x_j}\Theta^i_t(\phi_t(x)) = 
\sum_k \frac{\partial \Theta_t^i}{\partial x_k}(\phi_t(x))\, \frac{\partial \phi_t^k}{\partial x_j}(x),
\]
which says that
\begin{equation}\label{}
\dot K_t = \Theta_t' (\phi_t(x_0))\, K_t,\quad\text{where}\quad \Theta'_t = \begin{pmatrix}
\frac{\partial \Theta_t^i}{\partial x_j}
\end{pmatrix}
\end{equation}
is the Jacobian matrix of $\Theta_t:\bbR^{2n}\to\bbR^{2n}$ at $\phi_t(x_0)$.
Thus we obtain that $\dot K_t K_t^{-1} = \Theta'_t$, which then leads to:
\begin{equation}\label{aMedias}
\frac{d\ }{dt} \bbJ_t(\phi_t(x_0))  = [\Theta'_t(\phi_t(x_0))\,,\, \bbJ_t(\phi_t(x_0))].
\end{equation}
However
\[
\frac{d\ }{dt} \bbJ_t(\phi_t(x_0)) = \dot\bbJ_t(\phi_t(x_0)) + 
\sum_{k=1}^{2n}\Theta^k_t(\phi_t(x_0))\,\frac{\partial\bbJ_t}{x_k}(\phi_t(x_0)).
\]
We obtain:
\begin{lemma} At each $x\in M$ in the coordinate patch
\begin{equation}\label{}
\dot\bbJ_t(x) = [\Theta'_t(x)\,,\, \bbJ_t(x)] - 
\sum_{k=1}^{2n}\Theta^k_t(x)\,\frac{\partial\bbJ_t}{x_k}(x).
\end{equation}
\end{lemma}

\bigskip
So far we have not used the expression (\ref{Xi}) for $\Theta$.  
In  symplectic coordinates,
$(x_1, \ldots, x_n)= (p_1,\ldots,p_n,q_1,\ldots , q_n)$
\begin{equation}\label{}
\Theta_t = \Omega\,\nabla F + \bbJ_t\,\Omega\,\nabla G,
\end{equation}
where $
\Omega = 
\begin{pmatrix}
0 & -I\\
I & 0
\end{pmatrix}$ and we are using column vector notation.
Then 
\begin{equation}\label{}
\Theta'_t = \Omega\, F'' + \bbJ_t\,\Omega\,G'' + 
\begin{pmatrix}
\sum_{k=1}^{2n} \frac{\partial \Upsilon_{ik}}{\partial x_j} \gamma_k
\end{pmatrix}
\end{equation}
where we have introduced the notations
\[
F'' = \begin{pmatrix}
F_{pp} & F_{pq}\\ F_{qp} & F_{qq}
\end{pmatrix}
\quad\text{with}\quad
F_{pq} = \begin{pmatrix}
\frac{\partial^2 F}{\partial p_j\partial q_i}
\end{pmatrix},
\]  
where $i$ is the row index, etc., and
\[
 \quad \bbJ_t(x) = \begin{pmatrix} \Upsilon_{ij}(t,x)\end{pmatrix}, 
\quad (\gamma_1, \ldots ,\gamma_{2n}) = (-G_{q_1} , \cdots , -G_{q_n} ,,G_{p_1} , \cdots , G_{p_n}).
\]  
We have proved:
\begin{proposition}  Let $\Gamma $ be the matrix
\begin{equation}\label{}
\Gamma = \begin{pmatrix}
\sum_{k=1}^{2n} \frac{\partial \Upsilon_{ik}}{\partial x_j} \gamma_k
\end{pmatrix}.
\end{equation}
Under the evolution of the complex-valued Hamiltonian $h= F + iG: M\to\bbC$,
the complex structure evolves according to the equation
\begin{equation}\label{jDotEqCoords}
\dot\bbJ_t = [ \Omega F'' + \bbJ_t\Omega G''+\Gamma\,,\, \bbJ_t] 
-\sum_{k=1}^{2n}\Theta^k_t(x)\,\frac{\partial\bbJ_t}{x_k}(x).
\end{equation}
\end{proposition}

\medskip
In case $M=\bbR^{2n}$ and $F,\, G$ are quadratic forms,
$\bbJ_t$ is independent of the $x$ variables, and the 
equation (\ref{jDotEqCoords}) simplifies to
\begin{equation}\label{}
\dot\bbJ_t = [ \Omega F'' + \bbJ_t\Omega G''\,,\, \bbJ_t],
\end{equation}
which is in agreement with equation (48) in [Graefe-Schubert 2012]
(though the latter is written in terms of the metric
$\Omega\bbJ_t$).

\subsection{The case $M=\bbC^n$}

We consider $\bbR^{2n}$ with the standard symplectic structure and complex coordinates
$\zeta_j = \frac{1}{\sqrt{2}}(q_j+ip_j)$, so the symplectic structure is 
\begin{equation}\label{}
\omega = \sum_{j=1}^n dp_j\wedge dq_j = -i\sum_{j=1}^n d\zeta_j\wedge d\bar\zeta_j.
\end{equation}
Given $h:\bbR^{2n}\to \bbR$, its Hamilton field $\Xi_h$ is defined by the condition
$
\omega(\cdot, \Xi_h) = dh.
$
This gives the usual equations of motion
$\dot{q_j} = h_{p_j},\  \dot{p_j} = -h_{q_j}.$
One can check that in  complex coordinates Hamilton's equations are
\begin{equation}\label{haminCpx}
\dot\zeta_j = -i\frac{\partial h}{\partial \bar\zeta_j}
\end{equation}
and its complex conjugate (which is redundant). 

\bigskip
We complexify $\bbC^n$ by the anti-diagonal embedding
$\iota: \bbC^n\hookrightarrow \bbC^n\times\bbC^n$, 
$\zeta\mapsto (\zeta,\bar\zeta)$.  The initial projection $\Pi:\bbC^n\times\bbC\to\bbC^n$ is
just projection onto first factor.  We denote by $(z,w)$ complex variables on $\bbC^{2n}$.
The symplectic form $\omega$ extends analytically to the complex symplectic form
\[
\Omega = -i \sum_{j=1}^n dz_j\wedge dw_j.
\]
If $H:\bbC^{2n}\to\bbC$ is holomorphic, the associated Hamiltonian equations are:
\begin{equation}\label{cpxh}
\dot{z_j} = -iH_{w_j},\quad \dot{w_j} = iH_{z_j}.
\end{equation}

\bigskip
We now explain how to implement our scheme for constructing the exponential of 
$C^\omega(\bbR^{2n}, \bbC)$, where $\bbR^{2n}\cong\bbC^n$ as above.
Let $h:\bbR^{2n}\to\bbC$ be such that there is a holomorphic $H:\bbC^n\times\bbC^n\to\bbC$ such that 
\[
h = H\circ\iota.
\]   
Assume that Hamilton's
equations for $H$ can be integrated (take $t$ to be real, for simplicity), to yield a flow
\[
\Phi_t :\bbC^{2n}\to\bbC^{2n}.
\]
For each $\zeta\in\bbC^n$, let
\[
\calF^\zeta = \{ (\zeta, w)\;|\; w\in\bbC^n\}
\]
be the fiber over $\zeta$ of the projection $\Pi$.  Then $\phi_t(\zeta)\in\bbC^n$ is defined by the condition
\begin{equation}\label{theCond}
\{(\phi_t(\zeta),\overline{\phi_t(\zeta)})\}\in \Phi_t(\calF^\zeta).
\end{equation}

\bigskip
Now we proceed to examples, where we take $n=1$.


\medskip\noindent
\subsubsection{The imaginary harmonic oscillator}
This is  $h= \frac{i}{2}(q^2+p^2) = i\zeta\bar\zeta$, so that
$H = izw$.  Then
the equations are $\dot{z} = z$, $\dot{w} = -w$, so
\[
\Phi_t(z,w) = (e^{t}z, e^{-t} w).
\]
We must implement (\ref{theCond}) to find $\phi_t$.
In this case, this is trivial because Hamilton's equations separate:
\[
\{ (e^t\zeta, e^{-t}w)\;|\; w\in\bbC\} \cap \text{real locus} = \{(e^{t}\zeta, e^{t}\overline\zeta)\}.
\]
Therefore, in this case $\phi_t(\zeta) = e^{t}\zeta$, which is the gradient flow of $\Im(h)$, in agreement
with remark \ref{CompactGroup}.

\subsubsection{A quadratic, non-Hermitian example}\label{QuadNonHerm}
Let us take next an example that Graefe and Schubert also discuss in \cite{GrSc},
namely
\[
h = \frac{i}{2}q^2 = \frac{i}{4}(\zeta+\bar\zeta)^2.
\]
The analytic continuation of this Hamiltonian is just
\[
H =  \frac{i}{4}(z+w)^2, 
\]
and the equations of motion in the complexification are
\[
\dot{z} = -iH_w = \frac{1}{2}(z+w),\qquad \dot{w} = iH_z = -\frac{1}{2}(z+w).
\]
It is clear that $z+w$ is constant in time, and therefore the flow in the complexification is
\begin{equation}
\Psi_t(z,w) = \left(\frac{t}{2}(z+w) + z,  -\frac{t}{2}(z+w)+w\right).
\end{equation}
To find the induced map $\phi_t:\bbR^2\to\bbR^2$, we are to proceed as follows.  Fix $\zeta\in\bbC$, and 
flow-out under $\Phi_t$ the complex line
\begin{equation}
\calF^{\zeta} = \{(\zeta, w)\;;\; w\in\bbC\}.
\end{equation}
The result is
\begin{equation}
\calF^{\zeta}_t = \left\{\left(\frac{t}{2}(\zeta+w) + \zeta,  -\frac{t}{2}(\zeta+w)+w\right)\;;\; w\in\bbC\right\}.
\end{equation}
The points of intersection of this complex line with the real locus are given by the solutions
to the equation in $w$
\begin{equation}
 -\frac{t}{2}(\bar{\zeta}+\bar{w})+\bar{w} = \frac{t}{2}(\zeta+w) + \zeta ,
\end{equation}
or equivalently
\begin{equation}
\bar{w} - \zeta = t\left(\Re(\zeta)+\Re(w) \right).
\end{equation}
Let us write  $\zeta = a+ib$, $w=\alpha+i\beta$. The equation becomes the system
\begin{equation}
\alpha-a = t(\alpha+a), \quad
\beta =b.
\end{equation}
Assuming $t\not=1$ the solution is, $w=\frac{1+t}{1-t}a +ib$, which after some calculations yields
\begin{equation}
\phi_t(a+ib) = \frac{a}{1-t} + ib.
\end{equation}
As $t\to 1$ from the right, the image point tends to infinity.

\subsubsection{The generic element of the ``maximal torus"}
Let us now take $h = \varphi(|\zeta|^2)$ with $\varphi:\bbR\to\bbC$ having an analytic extension
$F: \bbC\to\bbC$.  Then
\[
H(z,w) = F(zw),
\]
and Hamilton's equations are 
\[
\dot{z} = -izF'(zw),\quad \dot{w} = iwF'(zw).
\]
Clearly, the function $zw$ is a constant of motion, and we can integrate:
\[
z(t) = e^{-itF'(\zeta w_0)}\zeta,\quad w(t) = e^{itF'(\zeta w_0)}w_0.
\]
Therefore  $\phi_t(\zeta) = e^{-itF'(\zeta w)}\zeta$ where $w$ solves
\begin{equation}\label{mustSolve}
e^{itF'(\zeta w)}w = \overline{\left(e^{-itF'(\zeta w)}\zeta\right)}.
\end{equation}
This is a transcendental equation.   However, we can easily prove:

\begin{lemma}  Fix $H$ as above.  Then $\phi_t$ commutes with the (usual) harmonic oscillator.
\end{lemma}
\begin{proof}
Fix $\zeta\in\bbC$, $w$ solving (\ref{mustSolve}), and $s\in\bbR$.  Let $\zeta_s = e^{-is}\zeta$
and $w_s = e^{is}w$.  Then
\[
e^{itF'(\zeta_s w_s)}w_s = e^{is}\,e^{itF'(\zeta w)} w = 
e^{is}\,\overline{\left(e^{-itF'(\zeta w)}\zeta\right)} =
\]
\[
=\overline{\left(e^{-itF'(\zeta w)}e^{-is} \zeta\right)} = 
\overline{\left(e^{-itF'(\zeta_s w_s)}\zeta_s\right)},
\]
which shows that $\phi_t(\zeta_s) = e^{-itF'(\zeta w)}\zeta_s = e^{-is}\phi_t(\zeta)$.

\end{proof}
\subsection{Coadjoint orbits}\label{CoAds}
We start by detailing the remark after Proposition \ref{conditions} above for the case of $M$ a coadjoint orbit of a compact Lie group $G_0$. For clarity, we restrict to the case of $G_0$ semi-simple. Let $\mathfrak{g}^*_0$ denote the dual of the Lie algebra $\frakg_0$ of $G_0$, and let $\lambda \in \frakg_0^*$. We identify $\frakg^*_0$ with $\frakg_0$ via the Killing form $B$. Thus there is a unique vector $\xi_{\lambda} \in \frakg_0$ such that $<\lambda, \eta> = B(\xi_{\lambda}, \eta)$, for all $\eta \in \frakg_0$. The Kostant-Kirilov $G_0$-invariant symplectic form on the orbit $G_0 \cdot \lambda := \mathcal{O}_{\lambda}$ is defined at $\lambda \in \mathcal{O}_{\lambda}$ via 
\[ \omega_{KK}(\tilde{\xi}, \tilde{\eta}) := \lambda([\xi, \eta]) = B(\xi_{\lambda}, [\xi, \eta]),\]
for all $\xi, \eta \in \frakg_0$, where $\tilde{\xi}$ is the vector field on $\mathcal{O}_{\lambda}$ induced by $\xi \in \frakg_0,$ etc. We will drop the subscript on the symplectic form. Let $H_0 \subset G_0$ be the isotropy group of $\lambda$, i.e., the centralizer of $\xi_{\lambda}$ and let $G$, resp. $H$ be the complexifications of $G_0$, resp., $H_0$. Note that any maximal torus $T_0$ in $H_0$ has $\xi_{\lambda}$ in its Lie algebra $\frakt_0$, and that $T_0$ is a maximal torus in $G_0$. Let $\Delta_{\lambda,+}$ be a set of positive roots for $\frakg$ for which $-i \alpha(\xi_{\lambda}) \geq 0,$ for all $\alpha \in \Delta_{\lambda,+}$ (i.e., for which $\lambda$ is in the closure of the corresponding Weyl chamber). Note that 
\[ \frakh = \frakh_0 \otimes \bbC = \frakt \; \oplus \; \oplus_{\{\alpha \, | \, <\alpha, \xi_{\lambda}> = 0\}} \frakg_{\alpha}.\]
Define the nilpotent algebra $${\frakn_{\lambda}}_{+} = \oplus_{\{\alpha \, | \, -i<\alpha, \xi_{\lambda}> > 0\}} \frakg_{\alpha},$$ and $N_{\lambda}$ the corresponding unipotent group in $G$. Let $P$ denote the parabolic subgroup of $G$ given by $H \cdot N_+$. The coadjoint orbit $\mathcal{O}_{\lambda} = G_0/H_0 = G/P$ has an induced complex structure from the second presentation. Two such structures coming from different sets of positive roots for $\frakg$ are equivalent, using the element of the Weyl group which relates these two sets of positive roots and fixes $\lambda$. There is an invariant Kaehler metric on $\mathcal{O}_{\lambda}$ such that the corresponding Kaehler form is the Kostant-Kirilov form.

\medskip

Now the complex orbit $\mathcal{O}_{\lambda, \bbC} = G \cdot \lambda \subset \frakg^*$ is a holomorphic symplectic manifold with symplectic form given by the Kostant-Kirilov prescription above, which is the complexification of the real $\mathcal{O}_{\lambda}, \omega$ above. The conjugation $\tau$ of $\frakg^*$ fixing $\frakg_0^*$ fixes $\mathcal{O}_{\lambda}$ and $\tau^* \omega_{\bbC} = \bar{\omega}_{\bbC}$. To complete the verification of the criteria (1) and (2) of Proposition \ref{conditions} above for $\mathcal{O}_{\lambda, \bbC}$, note that since $H \subset P$, the foliation of $G$ by left $P$-cosets is preserved under right multiplication by elements of $H$, and hence descends to give a foliation $\mathcal{F}$ of $\mathcal{O}_{\lambda, \bbC}$ which is invariant under the action of $G$. The leaves of this foliation are biholomorphic to $P/H \isom N_+ \isom \bbC^d$, where ``$\isom$" denotes isomorphism as algebraic varieties. That the leaves of this foliation are Lagrangian for $\omega_{KK, \bbC}$ amounts to the fact that $B(\xi_{\lambda}, [\xi, \eta]) = 0,$ which is because $ad(\xi_{\lambda}) \circ ad(\zeta)$ is nilpotent on $\frakg$ for any $\zeta \in \frakn_+$; for example, $\zeta = [\xi, \eta],$ for $\xi, \eta \in \frakn_+$. Finally, we have a canonical holomorphic $G$-equivariant mapping 
\[
\Pi: \mathcal{O}_{\lambda, \bbC} \to \mathcal{O}_{\lambda},\]
given by 
\[\mathcal{O}_{\lambda, \bbC} \ni gH \to gP \in G/P = \mathcal{O}_{\lambda},\]
whose fibers are obviously the leaves of the foliation $\mathcal{F}$.

\subsection{The case $M=\bbP^1$.}  We now specialize the discussion of coadjoint
orbits to the case of SU$(2)$.  
\subsubsection{Generalities}
We take the (skew-Hermitian) Pauli matrices as a
basis of the Lie algebra, 
\begin{equation}\label{sph4}
\sigma_1 = \frac{1}{2}
\begin{pmatrix}
0 & i \\
i & 0
\end{pmatrix}, \ 
\sigma_2 = \frac{1}{2}
\begin{pmatrix}
0 & -1 \\
1 & 0
\end{pmatrix}, \ 
\sigma_3 = \frac{1}{2}
\begin{pmatrix}
i & 0 \\
0 & -i
\end{pmatrix}
\end{equation}
 (so that $[\sigma_1,\sigma_2] = \sigma_3$, etc).   
  Give su$(2)$ the invariant inner product
such that the $\sigma_j$ are orthogonal and have length $1$, namely
\begin{equation}\label{ip}
\forall \alpha,\ \beta\in \text{su}(2)\qquad (\alpha, \beta) = -2\tr(\alpha\beta) = 2\tr(\alpha\beta^*),
\end{equation}
where $\beta^* = \overline{\beta}^T$, and use it to identify adjoint and coadjoint orbits.
 We will consider
 \[
 \calO = \text{adjoint orbit of }\sigma_3 \cong \bbC\bbP^1.
 \]
We let $x_k$ denote the $k$-th coordinate:
\[
x_k(\alpha) = (\alpha, \sigma_k).
\]
Explicitly,  if
\begin{equation}\label{realElt}
\alpha = \frac{i}{2}\begin{pmatrix}
a & z \\
\overline{z} & -a
\end{pmatrix} \in \text{su}(2)
\end{equation}
with $a\in\bbR,\ z\in\bbC$, then
\[
x_1(\alpha) =  \Re{z}, \quad
x_2(\alpha) = \Im{z}, \quad
x_3(\alpha) = a
\]
and $\calO$ is the unit sphere, $\sum x_j^2=1$.

\medskip
To complexify $\calO$ we introduce SL$(2,\bbC)$ and its Lie algebra.
We take the basis (\ref{sph4}) as a basis over $\bbC$ of sl$(2,\bbC)$.  
The quadratic form $-2\tr(\alpha\beta)$ is $G=$ SL$(2,\bbC)$ invariant and non-degenerate, so that we can 
continue to use it to identify adjoint and coadjoint orbits.

The SL$(2,\bbC)$ orbit through $\sigma_3$, $\calO_\bbC$, 
is the complexification of the previous orbit.
Let us denote a general element of sl$(2,\bbC)$ by
\begin{equation}\label{complexElt}
m = \frac{i}{2}\begin{pmatrix}
a & b\\
c & -a
\end{pmatrix},\quad a,\,b,\,c\in\bbC .
\end{equation}
Comparing with (\ref{realElt}), we see that the real locus is
$a \in\bbR$ and $ c = \bar{b}$.
The coordinate functions $x_j$ extend holomorphically to the functions
\begin{equation}\label{coordExtension}
X_1(m) = \frac{b+c}{2},\quad X_2(m) = \frac{b-c}{2i},\quad X_3(m) = a.
\end{equation}
The equation of the complex orbit, $\calO_\bbC$, is
\[
1 = \sum_j X_j^2 = bv + a^2 = 4\det(m),
\]
or $\det(m) = \frac 14$, which corresponds to $m$ having the eigenvalues $\pm \frac i2$.

\medskip
The unipotent group in SL$(2,\bbC)$ corresponding to $\sigma_3$ is
\begin{equation}\label{}
N_{+} = \left\{\begin{pmatrix} 1& n\\ 0 & 1\end{pmatrix}\;;\; n\in\bbC\right\},
\end{equation}
and the parabolic subgroup $P = TN_{+}$ is
\begin{equation}\label{}
P = \left\{\begin{pmatrix} t& n\\ 0 & t^{-1}\end{pmatrix}\;;\; n\in\bbC,\ t\in\bbC^*\right\}.
\end{equation}
If we now let
\begin{equation}\label{}
A = \left\{\begin{pmatrix} a& 0\\ 0 & a^{-1}\end{pmatrix}\;;\; a>0\right\},
\end{equation}
then we have  $T = T_0A$
(polar decomposition of complex numbers)
as well as the Iwasawa decomposition ($G=$SL$(2,\bbC)$, $G_0 = $SU$(2)$)
\begin{equation}\label{}
G = G_0AN
\end{equation}
and, therefore, $\displaystyle{G/P \cong G_0/T_0 }$.
More specifically, we have the commuting diagram of diffeomorphisms
\begin{equation}\label{}
\begin{array}{ccc}
G/P & \to & G_0/T_0\\
  & \searrow & \downarrow\\
  & & \calO
\end{array}
\end{equation}
where the top arrow is 
$gP\mapsto kT_0$ ($g=kan$ the Iwasawa decomposition of $g$),
the vertical arrow is $kT_0\mapsto k\cdot \sigma_3$

According to the general discussion of orbits, the projection $\Pi: \calO_{\bbC}\to\calO$
corresponds to the projection $G/T\to G/P$ given by $gT\mapsto gP$.  This shows that
\[
\Pi^{-1}(\sigma_3) = \{ gT\;;\; g\in AN\}.
\]
It is easy to see that this is the orbit of $N$ through $\sigma_3$, or, equivalently:
\begin{equation}\label{leafNorthPole}
\Pi^{-1}(\sigma_3) = \left\{ \frac 12 \begin{pmatrix} i & n \\ 0 & -i \end{pmatrix}\;;\; n\in\bbC
\right\}.
\end{equation}
This line is one of the two lines which are the intersection of the plane $X_3=1$ with the quadric.
By equivariance, we can conclude:

\begin{lemma}
For $\lambda\in\calO$, let $\ell_\lambda: \calO_\bbC \to\bbC$ be the function 
$\ell_\lambda(m) = \inner{m}{\lambda}$.  Then the fiber $\Pi^{-1}(\lambda)$ is 
one of the two lines whose union is $\ell_\lambda^{-1}(1)$.
\end{lemma}

\subsubsection{An alternate model}

There is an alternate model for the pair $(\calO,\, \calO_\bbC)$ in which the
fibers of the projection $\Pi: \calO_\bbC\to\calO$ are easier to describe.
Let us define
\begin{equation}\label{}
M := \{ (\ell\,, \ell^\bot)\;;\; \ell\subset\bbC^2\ \text{ 1-dimensional subspace}\, \}\cong\bbP^1
\end{equation}
where $\bbP^1$ is the complex projective line, and 
\begin{equation}\label{}
X:= \{ (\ell, m)\in\bbP^1\times\bbP^1\;;\; \ell\cap m = 0\}\cong
\left(\bbP^1\times\bbP^1\right)\setminus\bbP_\Delta,
\end{equation}
where $\bbP_\Delta$ is the diagonal.
We can identify $\calO_\bbC\cong X$ by
\begin{equation}\label{iden}
\calO_\bbC\ni A\mapsto (\ell_+\,,\,\ell_-)\in X, \quad \ell_{\pm} = \pm \frac i2\ \text{eigenspace of } A.
\end{equation}
Under this identification $\calO\subset\calO_\bbC$ gets identified with $M\subset X$.
\begin{lemma}
Under the  identification (\ref{iden}), the projection $\Pi: \calO_\bbC\to\calO$ is simply
\[
\Pi: X\to M\quad \Pi(\ell,\, m) = (\ell\,, \ell^\bot).
\]
Therefore, a leaf of the foliation of $\calO_\bbC$ consists of all elements
in $\calO_\bbC$ with a common $i/2$ eigenspace.
\end{lemma}
\begin{proof}
By (\ref{leafNorthPole}), the leaf through $\sigma_3$ consists  of the
elements in $\calO_\bbC$ having $e_1:=\langle 1, 0\rangle$ as an $i/2$ eigenvector.
Therefore, in our new model $\Pi: X\to M$
\[
\Pi^{-1}(\bbC e_1,\, \bbC e_2) = \{(\bbC e_1,\, m)\;;\; e_1\not\in m\}.
\]
The statement follows by SL$(2,\bbC)$ equivariance of the projection.
\end{proof}

\medskip
For future reference,
let us now use this lemma to find the leaf of the foliation of $\calO_\bbC$
consisting of all matrices having the vector
$\langle 1, \kappa\rangle$, $\kappa\in\bbC$, as an
$i/2$ eigenvector.  If $\kappa=0$ the leaf is just the leaf over $\sigma_3$, that is
(\ref{leafNorthPole}).  If $\kappa\not=0$, one can easily check that
\[
\frac i2\, \begin{pmatrix}
a & b \\ c & -a
\end{pmatrix}\,
\begin{pmatrix}
1 \\ \kappa
\end{pmatrix} = 
\frac i2 \begin{pmatrix}
1 \\ \kappa
\end{pmatrix}\quad \Leftrightarrow \quad
\begin{cases}
b = \frac{1-a}{\kappa}\\
c = \kappa(1+a).
\end{cases}
\]
Therefore, the leaf in question is
\begin{equation}\label{theLeaf}\renewcommand*{\arraystretch}{1.5}
\calL_\kappa := \left\{ \frac i2\,
\begin{pmatrix}\renewcommand*{\arraystretch}{1.5}
a & \frac{1-a}{\kappa}\\
 \kappa(1+a) & -a
\end{pmatrix}\;;\; a\in\bbC
\right\} .
\end{equation}
A calculation shows that 
the intersection of this leaf with the real locus $\calO$ is the matrix
\begin{equation}\label{}
A_\kappa := \frac {i}{2(1+|\kappa|^2)}\renewcommand*{\arraystretch}{1.5}
\begin{pmatrix}
1-|\kappa|^2 & 2\overline{\kappa}\\
2\kappa & |\kappa|^2-1
\end{pmatrix}.
\end{equation}
We see that $A_1 = \sigma_1$,  $A_{-i} = \sigma_2$ and $A_0=\sigma_3$.  
In fact the only
point in $\calO$ which is not of the form $A_\kappa$ for some $\kappa\in\bbC$
is $(-\sigma_3)$ (the only element in $\calO$ having $\langle 0, 1\rangle$
as $i/2$ eigenvector).  In fact, we can take $\kappa$ as a coordinate
on $\calO\setminus\{-\sigma_3\}$, centered at $\sigma_3$.

\subsubsection{An example of dynamics}
We can use the calculations above to find the trajectory of $A_\kappa$ under the
Hamiltonian $\frac{i}{2}x_3^2$.  Recall (\ref{coordExtension}), 
\[
x_3 \left[\frac i2\,\begin{pmatrix}
a & b \\ c & -a
\end{pmatrix}\right] = a.
\]
On the other hand the Hamilton flow of $x_3$ is conjugation by $\exp(t\sigma_3)$, that is
\[
\Psi_t\left[\frac i2\begin{pmatrix}
a & b \\ c & -a
\end{pmatrix}\right] = \frac i2
\begin{pmatrix}
e^{it/2} & 0 \\
0  & e^{-it/2}
\end{pmatrix}
\begin{pmatrix}
a & b \\ c & -a
\end{pmatrix}
\begin{pmatrix}
e^{-it/2} & 0 \\
0  & e^{it/2}
\end{pmatrix} =\frac i2
\begin{pmatrix}
a & e^{it}b \\ e^{-it}c & -a
\end{pmatrix}.
\]
To find the Hamilton flow of $\frac i2 x_3^2$ on $\calO_\bbC$ 
we simply replace $t$ by $itx_3 = ita$:
\begin{equation}\label{}
\Phi_t\left[\frac i2\begin{pmatrix}
a & b \\ c & -a
\end{pmatrix}\right] = \frac i2 \begin{pmatrix}
a & e^{-ta}b \\ e^{ta}c & -a
\end{pmatrix}.
\end{equation}
Let $\phi_t: M\to M$ be the exponential of $\frac i2 x_3^2$.
To find $\phi_t(A_\kappa)$ with $\kappa\not=0$
we look for the element $\varpi\in\calL_\kappa$ such that
$\Phi_t(\varpi)\in M$, and then $\phi_t(A_\kappa) = \Phi(\varpi)$.
Let $\varpi$ be as in the right-hand side of (\ref{theLeaf}), so that
\begin{equation}\label{}\renewcommand*{\arraystretch}{1.5}
\Phi_t(\varpi) = \frac i2
\begin{pmatrix}\renewcommand*{\arraystretch}{1.5}
a & e^{-ta}\,\frac{1-a}{\kappa}\\
e^{ta}\,\kappa(i+2a) & -a
\end{pmatrix}.
\end{equation}
For this matrix to be in $M$, i.e., for it to be skew-Hermitian, we need:
\begin{equation}\label{}
a  \in\bbR\quad\text{and}\quad
e^{ta}\,\kappa(i+2a)  =  \overline{
e^{-ta}\,\frac{1-a}{\kappa}} = e^{-ta}\,\frac{1-a}{\overline\kappa}.
\end{equation}
This is equivalent to
\begin{equation}\label{alphaEq}
e^{-2ta} ={|\kappa|^2}\,\frac{1+a}{1-a} 
\end{equation}
with $a\in\bbR$.
\begin{lemma}
If $\kappa = 0$  or $|\kappa|=1$, $A_\kappa$ is a fixed point of the exponential of $\frac i2 x_3^2$.
If $\kappa\not=0$, the exponential of $\frac i2 x_3^2$ with initial condition $A_\kappa$
exists for all time.  More specifically, $\forall t$  (\ref{alphaEq}) has a  unique (real) solution 
$a(t)$ which is smooth, and
\begin{equation}\label{}\renewcommand*{\arraystretch}{1.5}
\Phi_t(\varpi) = \frac i2
\begin{pmatrix}\renewcommand*{\arraystretch}{1.5}
a(t) & e^{-ta(t)}\,\frac{1-a(t)}{\kappa}\\
e^{ta(t)}\,\kappa(i+2a(t)) & -a(t)
\end{pmatrix}.
\end{equation}\end{lemma}
\begin{proof}
The function
\[
a\mapsto |\kappa|^2\,\frac{1+a}{1-a}
\]
is strictly increasing on $[-1\,,1)$ and maps this interval onto $[0, +\infty)$.
On the other hand, for any $t\in\bbR$ the function $a\mapsto e^{-2ta}$
is positive, monotone and bounded on $[-1\,,\,1]$.  
Therefore (\ref{alphaEq}) has a unique solution 
in $(-1,\, 1)$.  
Smoothness follows from the implicit function theorem:  $a(t)$ is implicitly
defined by the equation $F(t,a)=0$, where
\begin{equation}\label{implicitF}
F(t,a) = r^2\,\frac{1+a}{1-a}-e^{-2ta}
\end{equation}
with $r=|\kappa|$.  A calculation shows that, on $F=0$
\[
\frac{\partial F}{\partial a} = 
2\frac{r^2}{1-a}\left[ \frac{1}{1-a}
 + t(1+a)\right]
\]
where we have used the equation $F=0$ in the form $e^{-2ta} =r^2\,\frac{1+a}{1-a}$.
Therefore
\[
\frac{\partial F}{\partial a} = 0 \quad\Leftrightarrow \quad
t=\frac{1}{a^2-1},
\]
and therefore $\frac{\partial F}{\partial a} \not= 0$ for $a\in (-1, 1)$.
\end{proof}

The previous lemma says that the trajectories of the exponential of 
$\frac i2 x_3^2$ exist for all time and are smooth.  However, it is not true that,
for all $t$, $\phi_t:\calO\to\calO$ is a diffeomorphism.  We will show
this in the next section where we prove that in certain circumstances
(that include the present example), the complex structures $J_t$ must degenerate.  
We can also argue as follows.  Regard $a$, solving (\ref{alphaEq}), as a function 
of $t$ and $r=|\kappa|$.  By implicit differentiation with respect to $r$ this time,
one finds that
\begin{equation}\label{}
r\,  \left( \frac{1}{1-a^2} + t\right) \frac{\partial a}{\partial r} = -1.
\end{equation}
This equation cannot be satisfied if $t = \frac{1}{a^2-1}$.
For these values of $t$, $t\in(-\infty, -1]$, and, in fact, $\phi_t$ is a diffeomorphism
for $t\in(-1,\infty)$.

\medskip
We finally remark that, contrary to the example discussed
in \S \ref{QuadNonHerm} where $M=\bbC$, for $M=\calO$
it is not possible to have that $\Phi_t$ acts as holomorphic diffeomorphisms of 
the complexification for all time, and have a leaf of $\calF_t$ not intersect the real locus.
This is for topological reasons, as we now explain.
%
%
%
First, note that the complexification $X = \calO_\bbC$ is an affine quadric in $\bbC^3$.   
Each leaf of the foliation $\calF_0$ is a ruling complex line $\ell$ 
 of $\calO_\bbC$, which closes up to a projective line 
 $\bar{\ell} \subset \overline{\calO_\bbC} \subset \bbP^3.$ Suppose that, for
 some $t \in \bbC$ $\Phi_t(\ell)$, which is a leaf of $\calF_t$, does not intersect $M = S^2$. Then 
\[
	\Phi_{t}^{-1}(S^2) \cap \ell = \emptyset.
\]
Having compactified in $\bbP^3$, we have homology classes 
$[S^2] \in H_2(\calO_\bbC, \overline{\calO_\bbC}\setminus \calO_\bbC; \bbZ)$ 
and $[\bar{\ell}] \in H_2(\overline{\calO_\bbC}, \overline{\calO_\bbC}\setminus \calO_\bbC; \bbZ)$
and, calculating intersection products, we get
\[
	 [S^2]\cdot[\Phi_t\bar{\ell}] =[\Phi_{t}^{-1}(S^2)] \cdot [\bar{\ell}] =\pm 1 ,
\]
the sign depending on the choice of orientation on $S^2$. This contradicts $\Phi_t(\ell) \cap S^2 = \emptyset$.

\medskip
This argument is valid also for the complexifications of any coadjoint orbit of a compact group 
as described in \S \ref{CoAds}.

\section{Local obstructions to continuing a geodesic}


In this section, we examine obstructions to continuation of the geodesics given by $\phi_t$ 
generated by a purely imaginary Hamiltonian $iH, H$ real valued on $M$. If $H$ has a critical point 
at $x \in M$, then the holomorphic Hamiltonian flow $\Phi_t$ will have a fixed point at $x$. We will 
consider the transversality condition that $\calF_t$ intersect $\overline{\calF}_t$ transversally at 
$x$,and show that, in very simple examples, this transversality fails in finite time. We analyze the 
situation for $M$ of real dimension 2 and $H$ with non-zero Hessian completely. For higher 
dimensions, we resolve only the case of $H$ with a Hessian of rank 1, for example, $H = f^2$, 
where $f(x) = 0$ and $df(x) \neq 0$.


If $M$ is a real analytic manifold inside $X$ with conjugation $\sigma: X \to X$ having fixed points $M$, then, for $x \in M$, we can induce an action of $\sigma$ on $T_x^{(1,0)}X$ by $\tilde{\sigma}: v \to \overline{d\sigma_*(v)}.$ In standard local coordinates on $\bbC$, at $x \in M = \bbR \subset X = \bbC$, this just sends
%
\[
	v = a\frac{\partial \;\;}{\partial z} \to \bar{a}\frac{\partial \;\;}{\partial z}.
\]
Thus, we can identify the real tangent space $T_xM$ with the fixed points of $\tilde{\sigma}$ in $T_x^{(1,0)}X$. In other words, $v \in T_x^{(1,0)}X$ is in $T_xM$ if and only if $\;{\text{Re}}\; v \in T_xM \subset T_xX$, where these last are the real tangent spaces to the underlying real manifolds. The transversailty condition on $\calF_t$ becomes
\[
	T^{(1,0)}_x\calF_t \cap T_xM = \{0\} \; \text{in} \; T^{(1,0)}_xX. 
\]
%

If $dH(x) = 0$, we simply have $\Phi_t(x) \equiv x$, and $T_x^{(1,0)}\calF_t = d\Phi_{t,*}(T_x^{(1,0)}\calF_0) \subset T^{(1,0)}X$. In local Darboux coordinates $x_1, y_1, \ldots, x_n, y_n$, with $\omega =\sum dx_i \wedge dy_i,$, we have 
\[
	d\Phi_{t,*} = \exp(i \,t J \cdot \text{Hess}_x(H)). 
\]
%
%
We can now examine the transversality question purely locally at a critical point of $H$.
%
\begin{proposition} Let $x = 0 \in \bbR^2$ and $H$ be real analytic with $dH(0) = 0, \text{Hess}_0(H) \neq 0$. Let $\Phi_t$ be the holomorphic Hamiltonian flow of $i H$, which is defined in a sufficiently small neighborhood of $0$ for any given $t \in \bbC$. If $ \text{rank Hess}_0(H) = 2$, then $T^{(1,0)}\calF_0 \cap T_0\bbR^2 = \{0\},$ for all $t \in \bbC$. If $\text{rank Hess}_0(H) = 1$, then $T^{(1,0)}\calF_0 \cap T_0\bbR^2 \neq \{0\}$ for some finite $t \in \bbC$ (and in particular some $t \in \bbR$).
\end{proposition}
\begin{proof} In dimension 2 we can diagonalize $\text{Hess}_x(H)$ in Darboux coordinates, with eigenvalues $\lambda_1, \lambda_2 \in \bbR$.  There are three cases: $\lambda_1 \cdot \lambda_2 > 0, < 0, = 0.$ In local coordinates $x, y$ such that $\omega = dx \wedge dy,$ and 
%
\[
	\text{Hess}_x(H) = \left[\begin{array}{cc} \lambda_1 & 0 \\ 0 & \lambda_2 \end{array} \right].
\]
%
We also write $T_0^{(1,0)}\calF_0 = \bbC \cdot \left[\begin{array}{c} 1 \\ -i \end{array}\right].$ Assume first that $\mu := - \lambda_1 \cdot \lambda_2 > 0.$ Then
%
\[ \begin{array}{rcl}
	d\Phi_{t,*} & = & \exp(i t \left[\begin{array}{cc} 0 & - \lambda_2 \\ \lambda_1 & 0 \end{array} \right]) \\ && \\ & = & \cosh(\sqrt{\mu} t) I_2 + \frac{i}{\sqrt{\mu}}  \sinh(\sqrt{\mu} t) \left[\begin{array}{cc} 0 & - \lambda_2 \\ \lambda_1 & 0 \end{array} \right]. \end{array}
\]	
For $a, b \in \bbR$, $(a+ib)\cdot\left[\begin{array}{c} 1\\ - i\end{array}\right]$ is the general element of $T_0^{(1,0)}\calF_0$, and we seek solutions of $Im d\Phi_{0,*} ((a+ib)\cdot\left[\begin{array}{c} 1\\ - i\end{array}\right]) = 0.$ In matrix terms, we want
%
\[
	\left[\begin{array}{cc} \cosh(\sqrt{\mu} t) & - \frac{\lambda_2}{\sqrt{\mu}} \sinh(\sqrt{\mu}t) \\ & \\ \frac{\lambda_1}{\sqrt{\mu}} \sinh(\sqrt{\mu} t) & \cosh(\sqrt{\mu} t) \end{array}\right] \cdot \left[\begin{array}{c} a \\ \\ b \end{array}\right] = \left[\begin{array}{c} 0 \\ \\ 0\end{array}\right].
\]
%
%
But
%
%
\[
	\det \left[\begin{array}{cc} \cosh(\sqrt{\mu} t) & - \frac{\lambda_2}{\sqrt{\mu}} \sinh(\sqrt{\mu}t) \\ & \\ \frac{\lambda_1}{\sqrt{\mu}} 
	\sinh(\sqrt{\mu} t) & \cosh(\sqrt{\mu} t) \end{array}\right] = \cosh^2(\sqrt{\mu} t) + \frac{\lambda_1 \lambda_2}{\mu} \sinh^2(\sqrt{\mu} t) \equiv 1,
\] 
%
so that $T_0^{(1,0)}$ and $T_0\bbR^2$ are transverse for all $t \in \bbC$. Similar computations for $\nu := \lambda_1 \cdot \lambda_2 > 0$ lead to a system 
%
\[
	\left[\begin{array}{cc} \cos(\sqrt{\nu} t) & - \frac{\lambda_2}{\sqrt{\nu}} \sin(\sqrt{\nu}t) \\ & \\ \frac{\lambda_1}{\sqrt{\nu}} \sin(\sqrt{\nu} t) & \cos(\sqrt{\nu} t) \end{array}\right] \cdot \left[\begin{array}{c} a \\ \\ b \end{array}\right] = \left[\begin{array}{c} 0 \\ \\ 0\end{array}\right],
\]	
with 
\[
	\det \left[\begin{array}{cc} \cos(\sqrt{\nu} t) & - \frac{\lambda_2}{\sqrt{\nu}} \sin(\sqrt{\nu}t) \\ & \\ \frac{\lambda_1}{\sqrt{\nu}} \sin(\sqrt{\nu} t) & \cos(\sqrt{\nu} t) \end{array}\right] = \cos^2(\sqrt{\nu} t) + \frac{\lambda_1 \lambda_2}{\nu}\sin^2(\sqrt{\nu} t) \equiv 1,
\]
and, again, $T_0^{(1,0)}$ and $T_0\bbR^2$ are transverse for all $t \in \bbC$. Finally, if $\lambda_2 = 0$, then
\[
	d\Phi_{t,*} = \left[\begin{array}{cc} 1 & 0 \\ & \\ i t \lambda_1 & 1 \end{array}\right]
\]
and the system for transversality becomes
\[
	\left[\begin{array}{cc} 0 & 1 \\ & \\ t\lambda_1 - 1 & 0 \end{array}\right] \cdot \left[\begin{array}{c} a \\ \\ b \end{array}\right] = \left[\begin{array}{c} 0 \\ \\ 0\end{array}\right],
\]	
and
\[
	\det \left[\begin{array}{cc} 0 & 1 \\ & \\ t\lambda_1 - 1 & 0 \end{array}\right] = -(t \lambda_1 - 1),
\]
which is zero if and only if $t = \frac1{\lambda_1}.$  
\end{proof}


This results applies to $H = \frac{i}{2}x_3^2$ on $\calO\cong S^2$
as in the previous section and gives that the geodesic cannot be prolonged to $t = 1$.

The analysis above extends directly to higher dimensions for an $H$ with a critical point $x$ where the rank of $ \text{Hess}_x(H)$ is equal to one.

\end{document}